\documentclass[twocolumn,final,twoside,journal]{IEEEtran}

\usepackage{tabu}
\usepackage[usenames,dvipsnames]{xcolor}

\usepackage{amsmath,amsthm,graphicx,cite}
\usepackage{srcltx}
\usepackage{epsfig,amsfonts,subfigure}
\usepackage{graphicx,cite,amssymb,amsmath}
\usepackage{color}

\usepackage{amsthm}
\usepackage[nolist]{acronym}
\usepackage{psfrag}
\usepackage{perso}
\usepackage{pgfplots}
 \pgfplotsset{compat=newest}
    \pgfplotsset{plot coordinates/math parser=false}
    \pgfplotsset{
    label style={anchor=near ticklabel},
    xlabel style={yshift=0.0em},
    ylabel style={yshift=-0.3em},
    tick label style={font=\footnotesize },
    label style={font=\footnotesize},
    legend style={font=\footnotesize},
    title style={font=\fontsize{7}}}
\usepackage{xcolor}
\definecolor{iso}{rgb}{0.7,0.7,0.7}
\usepackage{blindtext}
\usepackage{flushend}
\usepackage{relsize}
\usepackage{mathtools}
\mathtoolsset{showonlyrefs}

\usepgflibrary{arrows}
\usetikzlibrary{patterns}
\usepgflibrary{decorations.pathmorphing}

\usepackage{soul}

\newcommand{\rosymb}{y} 

\newcommand{\absoverhead}{\delta}

\newcommand{\Pf } { \mathsf{P}_{\mathsf{F}}}

\newcommand{\barPf } { \bar {\mathsf{P}}_{\mathsf{F}}}

\newcommand{\Grx}{\tilde{\mathbf{G}}}

\ifthenelse{\isundefined{\Bmatrix}}{
  \newcommand{\Bmatrix}{\mathbf{B}}
}{
  \renewcommand{\Bmatrix}{\mathbf{B}}
}
\renewcommand{\Bmatrix}{\mathbf{B}}

\ifthenelse{\isundefined{\C}}{
    \newcommand{\C}[1]{\mathtt{C}_{#1}}
    }{
    \renewcommand{\C}[1]{\mathtt{C}_{#1}}
    }

\ifthenelse{\isundefined{\N}}{
    \newcommand{\N}[1]{\mathsf{N}_{#1}}
    }{
    \renewcommand{\N}[1]{\mathsf{N}_{#1}}
    }

\newcommand{\n}{\mathtt{n}}

\renewcommand{\nu}{\n_u}

\newcommand{\dmax}{ d_{\max}}

\newcommand{\x}{\mathtt{x}}

\ifthenelse{\isundefined{\m}}{
    \newcommand{\m}{m}
    }{
    \renewcommand{\m}{m}
    }

\newcommand{\Raptorinput}{u}
\newcommand{\vecu}{\mathbf{\Raptorinput}}

\newcommand{\Rintermsymbol}{v}
\newcommand{\vecv}{\mathbf{\Rintermsymbol}}

\newcommand{\Rosymb}{c}

\newcommand{\precodegeneric}{\mathcal{C}}

\newcommand{\GrxLT}{\Grx_{LT}}
\newcommand{\GLT}{\G_{\text{LT}}}

\newcommand{\Gp}{\G_{\text{o}}}

\newcommand{\Rrosymb}{\rosymb}

\newcommand{\we}{A}
\newcommand{\weo}{A}

\newcommand{\weoensemble}{\mathsf{A}}

\newcommand{\pil}{\pi_{\l}}

\renewcommand{\l}{l}

\ifthenelse{\isundefined{\g}}{
    \newcommand{\g}{\mathsf g}
    }{
    \renewcommand{\g}{\mathsf g}
    }

\ifthenelse{\isundefined{\G}}{
    \newcommand{\G}{\mathbf{G}}
    }{
    \renewcommand{\G}{\mathbf{G}}
    }

\newcommand{\Exp}{\mathbb{E}}

\newtheorem{theorem}{Theorem}

\newtheorem{lemma}{Lemma}

\ifthenelse{\isundefined{\A}}{
    \newcommand{\A}{\we}
    }{
    \renewcommand{\A}{\we}
    }

\ifthenelse{\isundefined{\argmax}}{
\newcommand{\argmax}{{\arg\max}}
}{
\renewcommand{\argmax}{{\arg\max}}
}

\newcommand{\krawt}{\mathcal{K}}

\newcommand{\rank}{\mathsf{rank}}

\newcommand{\pifroml}{\vartheta_{i,l,j} }
\newcommand{\qi}{\varphi_i }

\begin{document}
\begin{acronym}
\acro{BEC}{binary erasure channel}
\acro{DFT}{discrete Fourier transform}
\acro{$q$-EC}{$q$-ary erasure channel}
\acro{WE}{weight enumerator}
\acro{WEF}{weight enumerator function}
\acro{IOWEF}{input output weight enumerator function}
\acro{IOWE}{input output weight enumerator}
\acro{LT}{Luby Transform}
\acro{BP}{belief propagation}
\acro{ML}{maximum likelihood}
\acro{MDS}{maximum distance separable}
\acro{LDPC}{low density parity check}
\acro{i.i.d.}{independent and identically distributed}
\end{acronym}

\title{Bounds on the Error Probability of Raptor Codes}

\author{
    \IEEEauthorblockN{Francisco L\'azaro\IEEEauthorrefmark{1}, Gianluigi Liva\IEEEauthorrefmark{1}, Enrico Paolini\IEEEauthorrefmark{2}, Gerhard Bauch\IEEEauthorrefmark{3}}\\
    \IEEEauthorblockA{\IEEEauthorrefmark{1}Institute of Communications and Navigation of DLR (German Aerospace Center),
    \\Wessling, Germany. Email: \{Francisco.LazaroBlasco,Gianluigi.Liva\}@dlr.de}\\
    \IEEEauthorblockA{\IEEEauthorrefmark{2}{CNIT}, {DEI}, University of Bologna,
    \\Cesena, Italy. Email: e.paolini@unibo.it}\\
    \IEEEauthorblockA{\IEEEauthorrefmark{3}Institute for Telecommunication, Hamburg University of Technology
    \\Hamburg, Germany. Email: Bauch@tuhh.de}
    \thanks{This work has been accepted for publication at IEEE Globecom 2017.}
\thanks{\copyright 2016 IEEE. Personal use of this material is permitted. Permission
from IEEE must be obtained for all other uses, in any current or future media, including
reprinting /republishing this material for advertising or promotional purposes, creating new
collective works, for resale or redistribution to servers or lists, or reuse of any copyrighted
component of this work in other works}
}
\maketitle



\thispagestyle{empty} \pagestyle{empty}

\begin{abstract}
In this paper $q$-ary Raptor codes under \acs{ML} decoding are considered.
An upper bound on the probability of decoding failure is derived using the weight enumerator of the outer code, or its expected weight enumerator if the outer code is drawn randomly from some ensemble of codes. The bound is shown to be tight by means of simulations. This bound provides a new insight into Raptor codes since it shows how Raptor codes can be analyzed similarly to a classical fixed-rate serial concatenation.
\end{abstract}




\section{Introduction}\label{sec:Intro}

Fountain codes \cite{byers02:fountain} are a class of erasure codes that have the property of being rateless. Thus, they are potentially able to generate an endless amount of encoded (or output) symbols.  This property makes them suitable for application in situations where the channel erasure rate is not a priori known.
The first class of practical fountain codes, \ac{LT} codes, was introduced in \cite{luby02:LT} together with an iterative  decoding algorithm that achieves a good performance when the number of input symbols $k$ is large. In \cite{luby02:LT} it was shown how in order to achieve a low probability of decoding error, the encoding and iterative decoding cost per output symbol is $O \left(\ln(k)\right)$.

Raptor codes were introduced in \cite{shokrollahi06:raptor} and outperform \ac{LT} codes in many aspects. They consist of a serial concatenation of an outer code $\mathcal C$ (or precode) with an inner \ac{LT} code. On erasure channels, this construction allows relaxing the design of the \ac{LT} code, requiring only the recovery of a fraction $1-\gamma$ of the input symbols with $\gamma$ small. This can be achieved with linear encoding complexity and also linear decoding complexity using iterative decoding. The outer code is responsible for recovering the remaining fraction of input symbols, $\gamma$. If the outer code $\mathcal C$ is linear-time encodable and decodable then the Raptor code has linear encoding and iterative decoding complexity over erasure channels.

Most of the existing works on \ac{LT} and Raptor codes consider iterative decoding and assume large input block lengths ($k$ at least in the order of a few tens of thousands). However, in practice, smaller values of $k$ are more commonly used. For example, for the binary Raptor codes standardized in \cite{MBMS16:raptor} and \cite{luby2007rfc} the recommended values of $k$ range from $1024$ to $8192$. For these input block lengths, iterative decoding performance degrades considerably. In this context,
a different decoding algorithm is adopted that is an efficient \ac{ML} decoder, in the form of inactivation decoding \cite{shokrollahi2005systems}.

An inactivation decoder solves a system of equations in several stages. First a set of variables is declared \emph{inactive}. Next a system of equations involving the set of inactive variables needs to be solved, for example using Gaussian elimination. Finally, once the value of the inactive variables is known, all other variables are recovered using iterative  decoding.

Recently there have been several works addressing the complexity of inactivation decoding for Raptor and \ac{LT} codes \cite{lazaro:ITW,lazaro:scc2015,lazaro:Allerton2015,mahdaviani2012raptor}.
The probability of decoding failure of \ac{LT} and Raptor codes under \ac{ML} decoding has also been subject of study in several works. In \cite{Rahnavard:07}  upper and lower bounds to the intermediate symbol erasure rate were derived for \ac{LT} codes and Raptor codes with outer codes in which every element of the parity check matrix is \ac{i.i.d.} Bernoulli random variables with parameter $p$.
This work was extended in \cite{schotsch:2013}, where lower an upper bounds to the performance of \ac{LT} codes under \ac{ML} decoding were derived. A further extension was presented in \cite{Schotsch:14}, where an approximation to the performance of Raptor codes under \ac{ML} decoding is derived under the assumption that the number of erasures correctable by the outer code is small. Hence, this approximation holds only if the rate of the outer code is sufficiently high.
In \cite{Liva10:fountain} it was shown by means of simulations how the error probability of $q$-ary Raptor codes is very close to that of linear random fountain codes.
In \cite{wang:2015} upper and lower bounds to the probability of decoding failure of Raptor codes were derived. The outer codes considered in \cite{wang:2015} are binary linear random codes with a systematic encoder. Recently, ensembles of Raptor codes with linear random outer codes were also studied in a fixed-rate setting in \cite{lazaro:ISIT2015},\cite{lazaro:JSAC}. Although a number of works has studied the probability of decoding failure of Raptor codes, to the best of the knowledge of the authors, up to now the results hold only for specific binary outer codes (see \cite{Rahnavard:07,wang:2015,lazaro:ISIT2015,lazaro:JSAC}).

In this paper an upper bound on the probability of decoding failure of Raptor codes is derived, based on the weight enumerator of their outer codes. The bound is extended to ensembles of Raptor codes where the outer code is drawn randomly from an ensemble. In this case, it is necessary to know the average weight enumerator for the outer code ensemble. By means of simulations, the derived bound is shown to be tight, specially in the error floor region, for Raptor codes with Hamming and linear random outer codes. In contrast to \cite{Rahnavard:07,wang:2015,lazaro:ISIT2015,lazaro:JSAC} not only binary Raptor codes are considered, but also $q$-ary Raptor codes. The bounds presented in this paper can be seen as an extension of the upper bound in \cite{schotsch:2013} to Raptor codes.

The rest of the paper is organized as follows. In Section~\ref{sec:prelim} some preliminary definitions are presented. Section~\ref{sec:perf_bound} presents the upper bounds on the probability of decoding failure for the case in which the outer code is deterministic. In Section~\ref{sec:ensemble} these bounds are extended to the case in which the outer code is drawn from a linear parity-check based ensemble. Numerical results are presented in Section~\ref{sec:numres}. Section~\ref{sec:Conclusions} presents the conclusions of our work.

\section{Preliminaries}\label{sec:prelim}

We consider Raptor codes constructed over $\mathbb {F}_{q}$ with an $(h,k)$ outer linear block code $\precodegeneric$.
We shall denote the $k$ input (or source) symbols of a Raptor code as ${\vecu=(\Raptorinput_1,~\Raptorinput_2,~\ldots, \Raptorinput_k)}$. The elements of $\vecu$ belong to $\mathbb {F}_{q}$.
 Out of the $k$ input symbols, the outer code generates a vector of $h$ intermediate symbols ${\vecv=(\Rintermsymbol_1,~\Rintermsymbol_2,~\ldots, \Rintermsymbol_h)}  \in \precodegeneric$. Denoting by $\Gp$ the employed generator matrix of the outer code, of dimension $(k \times h)$ and with elements in $\mathbb {F}_{q}$, the intermediate symbols can be expressed as
\[
\vecv = \vecu \Gp.
\]
These intermediate symbols serve as input to an LT encoder, which can generate an unlimited number of output
symbols, ${\mathbf{\Rosymb}=(\Rosymb_1, \Rosymb_2, \ldots, \Rosymb_n)}$, where $n$ can grow unbounded. Again, the elements of $\mathbf{\Rosymb}$ belong to $\mathbb {F}_{q}$.
For any $n$ the output symbols can be expressed as
\[
\mathbf{\Rosymb} = \vecv \GLT = \vecu \Gp \GLT
\]
where $\GLT$ is an $(h \times n)$ matrix whose elements belong to  $\mathbb {F}_{q}$. Each column of $\GLT$ is associated with $\Rosymb_i$.
More specifically,
each column of $\GLT$ is generated by first selecting an output  degree $d$ according to the degree distribution ${\Omega= (\Omega_1, \Omega_2, \ldots, \Omega_{\dmax})}$, and then selecting $d$ different indexes uniformly at random between $1$ and $h$. Finally, the elements of the column corresponding to these indexes are drawn independently and uniformly at random from $\mathbb {F}_{q} \backslash \{0\}$, while all other elements of the column are set to zero.

The output symbols $\mathbf{\Rosymb}$ are transmitted over a \ac{$q$-EC} at the output of which each transmitted symbol is either correctly received or erased.\footnote{The results developed in this paper remain valid regardless the statistic of the erasures introduced by the channel.}
We denote by $m$ the number of output symbols collected by the receiver of interest, and we express it as $m=k+\absoverhead$. Let us denote by ${\mathbf{\Rrosymb}=(\Rrosymb_1, \Rrosymb_2, \ldots, \Rrosymb_m)}$ the $m$ received output symbols. Denoting by $\mathcal{I} = \{i_1, i_2, \hdots, i_m \}$ the set of indices corresponding to the $m$ non-erased symbols, we have
\[
\Rrosymb_j = \Rosymb_{i_j}.
\]
An \ac{ML} decoder (for example, an inactivation decoder) proceeds by solving the linear system of equations
\[
\mathbf{\Rrosymb} = \vecu \Grx
\]
where
\begin{align}
\Grx = \Gp \GrxLT
\label{eq:sys_eq}
\end{align}
with $\GrxLT$ given by the $m$ columns of $\GLT$ with indices in $\mathcal{I}$.

Given a block code $\precodegeneric$ of length $h$ we shall denote its weight enumerator as $\weo = \{\weo_0, \weo_1 \hdots \weo_h\}$, where $\weo_i$ denotes the multiplicity of codewords of weight $i$.
Similarly, given an ensemble of block codes, all with the same length $h$, along with a probability distribution on the codes in the ensmble, we shall denote its average weight enumerator as ${\weoensemble = \{\weoensemble_0, \weoensemble_1 \hdots \weoensemble_h\}}$, where $\weoensemble_i$ denotes the expected multiplicity of codewords of weight $i$ of a code drawn randomly from the ensemble.

\section{Upper Bounds on the Error Probability}\label{sec:perf_bound}

The following theorem establishes an upper bound on the probability of decoding failure $\Pf$ under \ac{ML} decoding of a Raptor code constructed over $\mathbb {F}_{q}$ as a function of the receiver overhead $\absoverhead$.

\begin{theorem}\label{theorem:rateless}
Consider a Raptor code constructed over $\mathbb {F}_{q}$ with an $(h,k)$ outer code $\precodegeneric$ characterized by a weight enumerator $\weo$, and an inner \ac{LT} code with output degree distribution $\Omega$.
The probability of decoding failure under optimum erasure decoding given that ${m=k+\absoverhead}$ output symbols have been collected by the receiver can be upper bounded as
\[
\Pf  \leq \sum_{l=1}^h \weo_{\l} \pil^{k+\absoverhead}
\]
where  $\pil$ is the probability that a generic output symbol is equal to $0$ given that the vector $\vecv$ of intermediate symbols has Hamming weight $l$. The expression of $\pil$ is \cite{schotsch:2013}
\begin{align}
\pil &= \frac{1}{q} +  \frac{q-1}{q} \sum_{j=1}^{\dmax} \Omega_j    \frac{\krawt_j(l; h,q)}{\krawt_j(0; h,q)}
\label{eq:pl}
\end{align}
where $\krawt_j(l; h,q)$ is the Krawtchouk polynomial of degree $j$ with parameters $h$ and $q$.\footnote{The Krawtchouk polynomial of degree $j$ with parameters $n$ and $q$  is defined as \cite{MacWillimas77:Book}
\[
\krawt_k(x;n,q) = \sum_{j=0}^k (-1)^j \binom{x}{j} \binom{n-x}{k-j} (q-1)^{k-j}.
\]
}
\end{theorem}
\begin{proof}
An optimum (e.g. inactivation) decoder solves the linear system of equations in \eqref{eq:sys_eq}. Decoding fails whenever the system does not admit a unique solution, that is, if and only if  $\rank(\Grx)<k$, i.e. if
${\exists\,  \vecu \in \mathbb {F}_q^k \backslash \{ \textbf{0}\} \,\, \text{s.t.} \,\, \vecu \Grx = \textbf{0}}$.
Consider two  vectors $\vecu \in \mathbb {F}_q^k, \vecv \in \mathbb {F}_q^h$. Let us define  $E_{\vecu}$ as the event $\vecu \mathbf{G}_o \GrxLT = \mathbf{0}$.
Similarly, we define $E_{\vecv}$ as the event  $\vecv \GrxLT = \mathbf{0}$. We have
\begin{align}
\Pf & =   \Pr\left\{ \small{\bigcup_{\vecu \in \mathbb {F}_q^k \backslash \{ \textbf{0}\}}} E_{\vecu}  \right\}  = \Pr\left\{ \small{\bigcup_{\vecv \in \precodegeneric \backslash \{ \textbf{0}\} }} E_{\vecv} \right\}
\label{eq:existence}
\end{align}
where we made use of the fact that due to linearity, the all zero intermediate word is only generated by the all zero input vector.

Developing \eqref{eq:existence} we have
\begin{align}
\Pf & = \Pr \left\{ \small{\bigcup_{l=1}^h} \,\, \small{ \bigcup_{\vecv \in  \mathbb \precodegeneric_l } } E_{\vecv}  \right\}
\label{eq:existence2}
\end{align}
where, by definition
\[
\precodegeneric_l = \left\{ \vecv \in \precodegeneric : w_H(\vecv) = l \right\}
\]
is the set of codewords in $\precodegeneric$ of Hamming weight $l$.

Let $L$ be a discrete random variable representing the Hamming weight of vector $\vecv \in \precodegeneric$. Moreover, let $J$ and $I$ be discrete random variables representing the number of intermediate symbols which are linearly combined to generate the generic output symbol $y$, and the number of non-zero such intermediate symbols, respectively. Note that $I \leq L$. We can upper bound \eqref{eq:existence2} as
\begin{align}
\Pf & \leq \sum_{l=1}^{h} \Pr \left\{ \small{ \bigcup_{\vecv \in  \mathbb \precodegeneric_l } } E_{\vecv}  \right\} 
\leq \sum_{l=1}^{h} \weo_{\l} \Pr \left\{ E_{\vecv} | L=l  \right\} \, .
\label{eq:existence3}
\end{align}
Observing that the output symbols are independent of each other, we have
\[
\Pr \left\{ E_{\vecv} | L=l \right\} = \pil^{k+\absoverhead}
\]
where $\pil = \Pr \{ y=0 | L=l\}$. An expression for $\pil$ may be obtained observing that
\begin{align*}
\pil &= \sum_{j=1}^{\dmax} \Pr \{ y=0 | L=l,J=j \} \Pr \{ J=j | L=l \} \\
      &\stackrel{(\mathrm{a})}{=} \sum_{j=1}^{\dmax} \Omega_j \Pr \{ y=0 | L=l,J=j \} \\
      &\stackrel{(\mathrm{b})}{=} \sum_{j=1}^{\dmax} \Omega_j \sum_{i=0}^{\min\{j,l\}} \Pr \{ y=0 | I=i \} \! \Pr \{ I=i | L=l, J=j \}
\end{align*}
where equality `$(\mathrm{a})$' is due to \mbox{$\Pr \{ J=j | L=l \} = \Pr \{ J=j \}$} $= \Omega_j$ and equality `$(\mathrm{b})$' to $\Pr \{ y=0 | L=l, J=j, I=i \} = \Pr \{ y=0 | I=i \}$. Letting \mbox{$\pifroml = \Pr \{ I=i | L=l, J=j \}$}, since the $j$ intermediate symbols are chosen uniformly at random by the LT encoder we have
\begin{align}\label{eq:neighbors}
\pifroml = \frac{ \binom{\l}{i} \binom{h-\l}{j-i} } { \binom{h}{j}} \, .
\end{align}
Let us denote $\Pr \{ y=0 | I=i \}$ by $\qi$ and let us observe that, due to the elements of $\Grx$ being \ac{i.i.d.} and uniformly drawn in $\mathbb {F}_{q} \setminus \{0\}$, on invoking Lemma~\ref{lemma:galois} in the Appendix\footnote{The proof in the Appendix is only valid for fields with characteristic $2$, the case of most interest for practical purposes. The proof of the general case is a trivial extension of Lemma~\ref{lemma:galois} in the Appendix.} we have
\begin{align}\label{eq:sum}
\qi =\frac{1}{q} \left( 1 + \frac{(-1)^i}{(q-1)^{i-1}}\right).
\end{align}
We conclude that $\pil$ is given by
\begin{align}
\pil &=  \sum_{j=1}^{\dmax} \Omega_j  \sum_{i=0}^{\min\{j,l\}} \pifroml  \, \qi
\end{align}
where $\pifroml$ and $\qi$ are given by \eqref{eq:neighbors} and \eqref{eq:sum}, respectively.

Expanding this expression and rewriting it using Krawtchouk polynomials and making use of the Chu-Vandermonde identity, one obtains \eqref{eq:pl}. This completes the proof.
\end{proof}

The following theorem makes the bound in Theorem~\ref{theorem:rateless} tighter for $q>2$. It is equivalent to Theorem~\ref{theorem:rateless} for $q=2$.
\begin{theorem}\label{lemma:bound_tight}
Consider a Raptor code constructed over $\mathbb {F}_{q}$ with an $(h,k)$ outer code $\precodegeneric$ characterized by a weight enumerator $\weo$, and an inner \ac{LT} with output degree distribution $\Omega$.
The probability of decoding failure under optimum erasure decoding given that ${m=k+\absoverhead}$ output symbols have been collected by the receiver can be upper bounded as
\[
\Pf  \leq \sum_{l=1}^h \frac{\weo_{\l}}{q-1} \pil^{k+\absoverhead}
\]
\end{theorem}
\begin{proof}
The bound \eqref{eq:existence3} can be tightened by a factor $q-1$ exploiting the fact that for a linear block code $\precodegeneric$ constructed over $\mathbb {F}_{q}$, if $\mathbf{\Rosymb}$ is a codeword, $\alpha \mathbf{\Rosymb}$ is also a codeword, $\forall \alpha \in \mathbb F_{q} \backslash \{0\}$ \cite{Liva2013}.
\end{proof}
The upper bound in Theorem~\ref{lemma:bound_tight} also applies to LT codes. In that case, $\weo_{\l}$ is simply the total number of sequences of Hamming weight $l$ and length $k$,
\[
\weo_{\l}= \binom{k}{l} (q-1)^{l-1}.
\]
The upper bound obtained for LT codes coincides with the bound in \cite{schotsch:2013} (Theorem 1). 

\section{Case of Random Outer Codes from Linear Parity-Check Based Ensembles}\label{sec:ensemble}

Both Theorem~\ref{theorem:rateless} and Theorem~\ref{lemma:bound_tight} apply to the case of a specific outer code. Next we extend these results to the case of a random outer code drawn from an ensemble of codes. Specifically, we consider a parity-check based ensemble of outer codes, denoted by $\msr{C}^\text{o}$, defined by a random matrix of size $(h - k) \times h$ whose elements belong to $\mathbb F_q$. A linear block code of length $h$ belongs to $\msr{C}^\text{o}$ if and only if at least one of the instances of the random matrix is a valid parity-check matrix for it. Moreover, the probability measure of each code in the ensemble is the sum of the probabilities of all instances of the random matrix which are valid parity-check matrices for that code. Note that all codes in $\msr{C}^\text{o}$ are linear, have length $h$, and have dimension $k_\precodegeneric \geq k$.
In the following we use the expression ``Raptor code ensemble'' to refer to the set of Raptor codes obtained by concatenating an outer code belonging to the 
ensemble $\msr{C}^\text{o}$ with an \ac{LT} encoder having distribution $\Omega$. We shall denote this ensemble as $(\msr{C}^\text{o}, \Omega)$.

\begin{theorem}\label{corollary:rateless}
Consider a Raptor code ensemble $(\msr{C}^o, \Omega)$ and let $\weoensemble= \{ \weoensemble_0,\weoensemble_1,\dots,\weoensemble_h \}$ be the expected weight enumerator of
a  code $\precodegeneric$ that is randomly drawn from  $\msr{C}^o$, i.e., let ${\weoensemble_{l} = \Exp_{ \msr{C}^o }[A_l(\precodegeneric)]}$ for all $l \in \{0,1,\dots,h\}$. Let
\begin{align}
\barPf = \Exp_{  \msr{C}^o } [ \Pf(\precodegeneric)]
\label{eq:ensemble}
\end{align}
be the average probability of decoding failure of the Raptor code obtained by concatenating an instance of $\precodegeneric$ with the \ac{LT} encoder, under optimum erasure decoding and given that ${m=k+\absoverhead}$ output symbols have been collected by the receiver. Then
\[
\barPf  \leq  \sum_{l=1}^h \frac{\weoensemble_{\l}}{q-1}  \pil^{k+\absoverhead} \, .
\]
\end{theorem}
\begin{proof}
Due to Theorem~\ref{lemma:bound_tight}  we may write
\begin{align}
\barPf \leq \Exp_{ \msr{C}^o } \left[ \sum_{l=1}^h \frac{\weo_{\l}(\precodegeneric) }{q-1} \pil^{k_\precodegeneric+\absoverhead} \right].
\label{eq:ensemble2}
\end{align}

For all  outer codes $\precodegeneric \in \msr{C}^\text{o}$ we have $k_\precodegeneric \geq k$. Since $\pil \leq 1$  we can write
\[
\pil^{k_\precodegeneric+\absoverhead} \leq \pil^{k+\absoverhead}
\]
which allows us to upper bound \eqref{eq:ensemble2} as
\[
\barPf \leq \Exp_{  \msr{C}^o } \left[ \sum_{l=1}^h \frac{\weo_{\l}(\precodegeneric) }{q-1} \pil^{k+\absoverhead} \right]= \sum_{l=1}^h \frac{\weoensemble_{\l}}{q-1}  \pil^{k+\absoverhead}
\]
where the last equality follows from linearity of expectation.
\end{proof}

\section{Numerical Results}\label{sec:numres}

All  results presented in this section use the \ac{LT} output degree distribution employed by standard R10 Raptor codes,  \cite{MBMS16:raptor,luby2007rfc},
\begin{align}
\Omega(\x) &= \sum_{j=1}^{\dmax} \Omega_j \x^j  \\
&= 0.0098\x + 0.4590\x^2+ 0.2110\x^3+0.1134\x^4 \\
&+ 0.1113\x^{10} + 0.0799\x^{11} + 0.0156\x^{40}.
\label{eq:dist_mbms}
\end{align}

\subsection{Binary Raptor Codes with Hamming Outer Codes}
In this section we consider binary Raptor codes with (deterministically known) Hamming outer codes.
The weight enumerator of a binary Hamming code of length $h=2^t-1$ and dimension $k=h-t$ can be derived easily using the  recursion
\[
(i+1)\, A_{i+1} + A_i + (h-i+1)\, A_{i-1}= \binom{h}{i}
\]
with $A_0=1$ and $A_1=0$ \cite{MacWillimas77:Book}. The weight distribution obtained from this recursion can then incorporated in Theorem~\ref{theorem:rateless} to derive the corresponding upper bound on the probability of Raptor decoding failure under optimum decoding.

\begin{figure}[t]
    \centering
    \includegraphics[height=7.8cm]{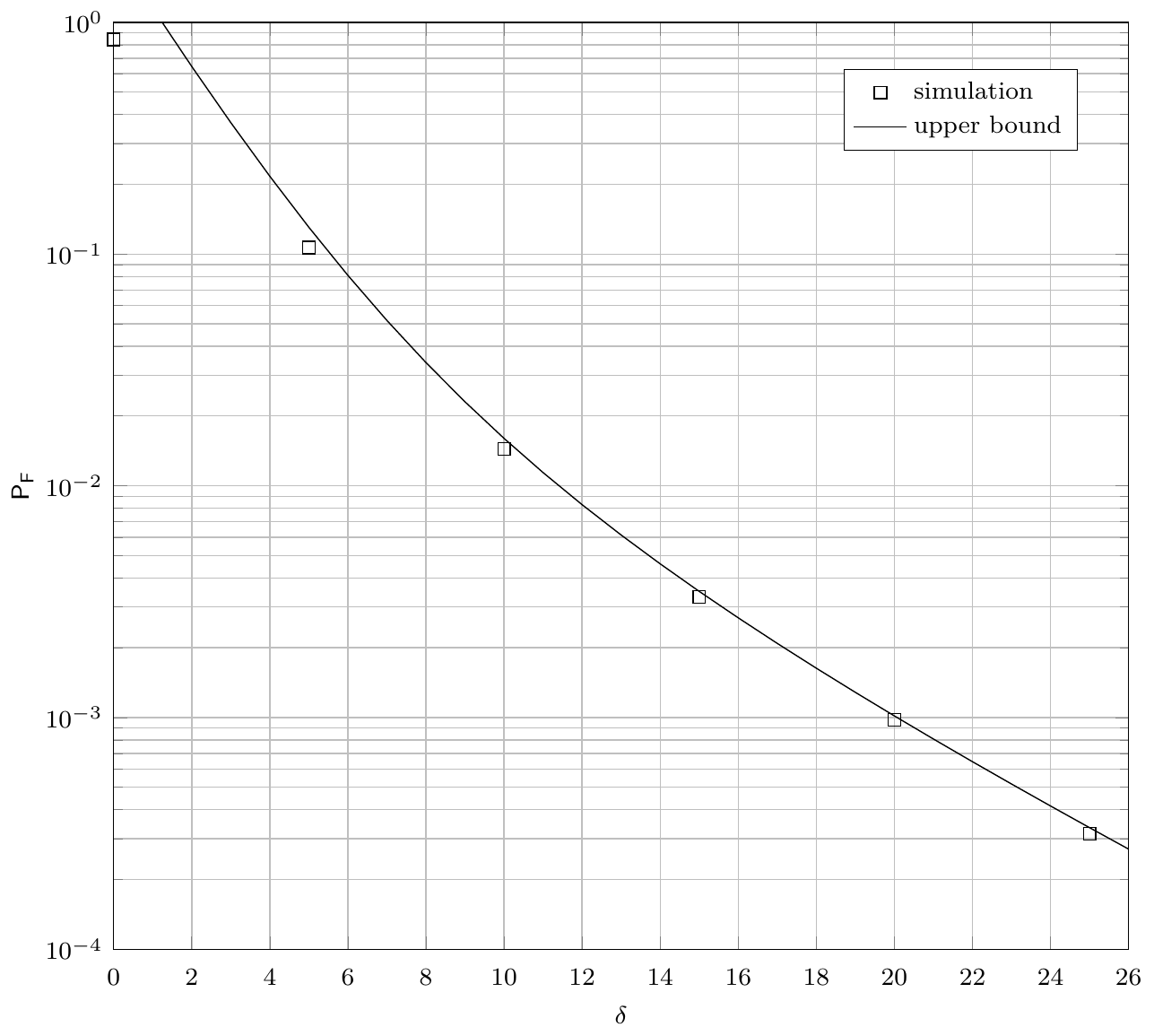}
    \caption[Probability of decoding failure $\Pf$ vs the absolute overhead for a Raptor with a $(63,57)$ Hamming outer code]{Probability of decoding failure $\Pf$ versus the absolute overhead for a binary Raptor code with a $(63,57)$ Hamming outer code. The solid line denotes the upper bound on the probability of decoding failure expressed  by Theorem~\ref{theorem:rateless}. The markers denote simulation results.}\label{fig:Hamming_sim}
\end{figure}

Figure~\ref{fig:Hamming_sim} shows the decoding failure rate for a binary Raptor code using a $(63,57)$ binary Hamming outer code as a function of the absolute overhead, $\absoverhead$.  The upper bound established in Theorem 1 is also shown.
In order to obtain the values of failure rate, for each $\absoverhead$ value Monte Carlo simulations were run until $200$ errors were collected using inactivation decoding. It can be observed how the upper bound is tight.

\subsection{Linear Random Outer Code}
In this subsection, we consider a $(\msr{C}^o, \Omega)$  Raptor code ensemble constructed over $\mathbb F_q$, where the \ac{LT} distribution $\Omega$ is the one defined in \eqref{eq:dist_mbms} and where $\msr{C}^o$ is the uniform parity-check ensemble, with
  parity-check matrix of size $(h-k) \times h$ and characterized by \ac{i.i.d.} entries with uniform distribution in $\mathbb F_q$. The expected multiplicity of codewords of weight $l$ for an outer code drawn randomly in $\msr{C}^o$ according to the described procedure is known to be
\begin{align}
\weoensemble_{\l} = \binom{h}{\l} q^{-(h-k)}  (q-1)^l.
\label{eq:wef_random}
\end{align}

In order to obtain the experimental values of decoding failure rate, $6000$ different outer codes were generated. For each outer code and for each overhead value $10^3$ inactivation decoding attempts were carried out. The average  failure rate was calculated by  averaging the failure rates of the individual Raptor codes.
In order to select the outer code an $(h-k)\times h$ parity check matrix was selected at random by generating each of its elements according to a uniform distribution in $\mathbb F_{q}$.

In Figure~\ref{fig:pf_k_64_h70} we show simulation results for $k=64$ and $h=70$. Two different $(\msr{C}^o, \Omega)$ Raptor code ensembles were considered, one constructed over $\mathbb F_{2}$ and one constructed over $\mathbb{F}_4$. We can observe how in both cases the bounds hold and are tight except for very small values of $\absoverhead$.
\begin{figure}[th]
    \centering
    \includegraphics[height=7.8cm]{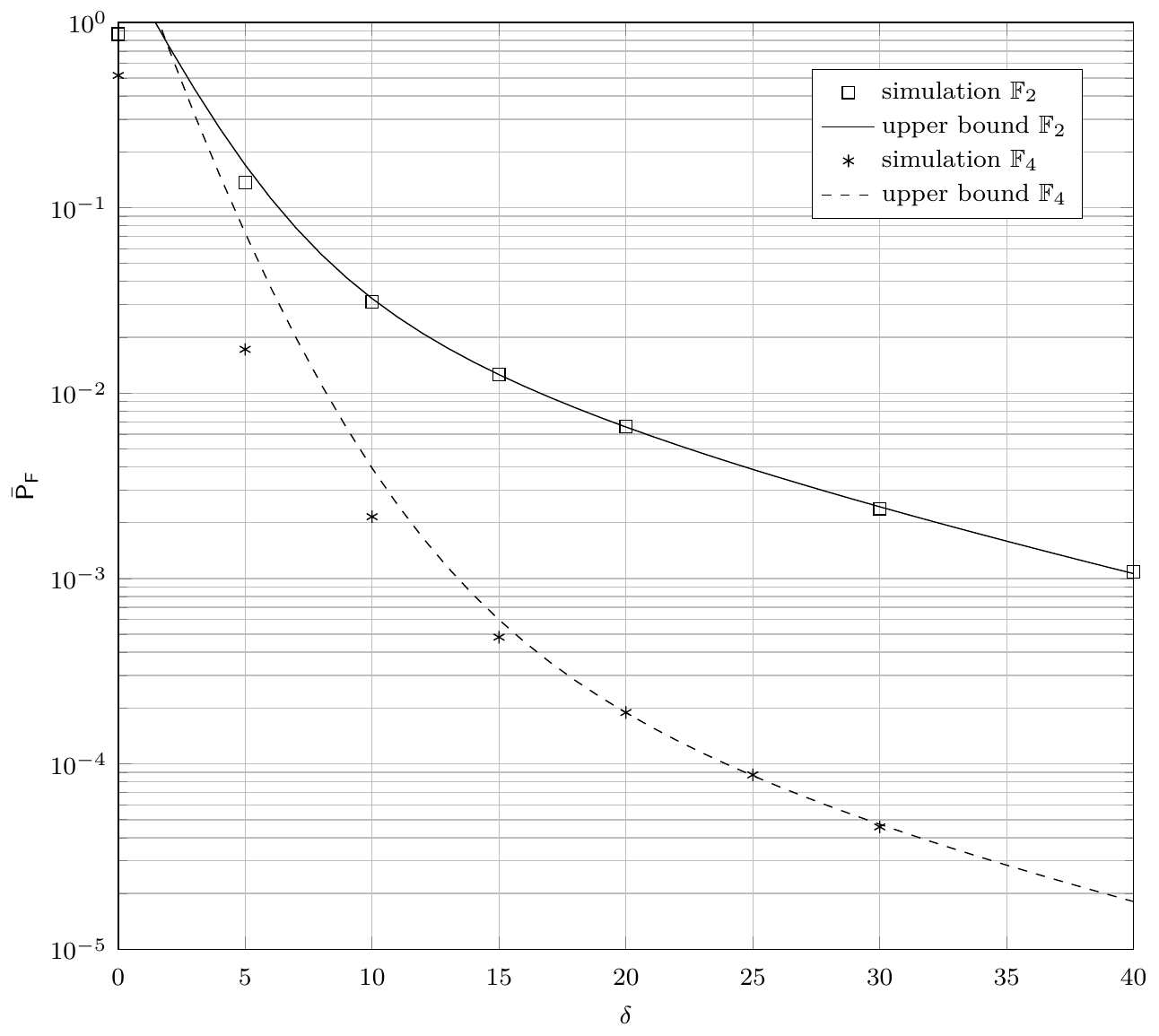}
    \caption{Expected probability of decoding failure $\barPf$ vs absolute overhead for Raptor code ensembles where the outer code is drawn randomly from the uniform parity-check ensemble. The solid and dashed lines denote the upper bounds on the average probability of decoding failure for
    the ensembles constructed over $\mathbb{F}_2$ and $\mathbb{F}_4$ respectively.
    The  markers denote simulation results.}\label{fig:pf_k_64_h70}
\end{figure}

\section{Conclusions}\label{sec:Conclusions}
In this paper we have consider Raptor codes under \acs{ML} decoding.
We have derived an upper bound on the probability of decoding failure of Raptor codes with generic  $q$-ary outer codes. This bound is general and only requires the knowledge of the weight enumerator of the outer code. The bound also applies to ensembles of Raptor codes where the outer code is randomly selected from an ensemble.
The bound is shown to be tight, specially in the error floor, by means of simulations.


\appendix \label{sec:appendix}
The following lemma is used in the proof of Theorem~\ref{theorem:rateless}.

\begin{lemma}\label{lemma:galois}
Let $X_1$, $X_2$ ... $X_l$  be discrete i.i.d random variables uniformly distributed over $\mathbb F_{2^m} \backslash \{0\}$. Then
\[
\Pr \{X_1 + X_2+ \hdots + X_l = 0 \}= \frac{1}{q} \left( 1 + \frac{(-1)^i}{(q-1)^{i-1}}\right)
\]
where $q=2^m$.
\end{lemma}

\begin{proof}
Observe that the additive group of $\mathbb F_{2^m}$ is isomorphic to the vector space $\mathbb Z_2^m$. Thus, we may let
$X_1$, $X_2$ ... $X_l$  be i.i.d random variables with uniform probability mass function over the vector space $\mathbb Z_2^m \backslash \{0\}$.

Let us introduce the auxiliary random variable
\[
W := X_1 + X_2+ \hdots + X_l
\]
and let us denote by $P_W(w)$ and by $P_X(x)$ the probability mass functions of $W$ and $X_i$, respectively, where
\[
P_X(x) =
\begin{cases}
 0 & \text{if } x=0 \\
 \frac{1}{q-1} & \text{otherwise.}
 \end{cases}
\]
Due to independence we have
\[
P_W = P_X \ast  P_X \ast \hdots \ast P_X
\]
which, taking the $m$-dimensional  two-points \ac{DFT} $\msr{J} \{\cdot\}$ of both sides, yields

\[
\msr J \{ P_W(w) \} =  \left( \msr J \{ P_X(x) \} \right)^l.
\]
Next, since
\[
\hat P_X(t) := \msr J \{ P_X(x) \}=
\begin{cases}
 1 & \text{if } t=0 \\
 \frac{-1}{q-1} & \text{otherwise}
 \end{cases}
\]
we have
\[
\hat P_W(t) := \msr J \{ P_W(w) \} =
\begin{cases}
 1 & \text{if } t=0 \\
 \frac{(-1)^l}{(q-1)^l} & \text{otherwise.}
 \end{cases}
\]
We are interested in  $P_W(0)$ whose expression corresponds to
\begin{equation}
 P_W(0) = \frac{1}{q} \sum_t \hat P_W(t) = \frac{1}{q} + \frac{1}{q} (q-1)  \frac{(-1)^l}{(q-1)^l}
\end{equation}
from which the statement follows.
\end{proof}

The result in this lemma appears in \cite{schotsch:2013}. However, the proof in \cite{schotsch:2013} uses a different approach based on a known result on the number of closed walks of length $l$ in a complete graph of size $q$ from a fixed but arbitrary vertex back to itself.

\section*{Acknowledgements}
This work was supported in part by {ESA/ESTEC} under Contract No. {4000111690/14/NL/FE} ``NEXCODE''.

\bibliographystyle{IEEEtran}
\bibliography{IEEEabrv,references}

\end{document}